\documentclass[prl,10pt,twocolumn,notitlepage]{revtex4-1}
\usepackage[margin=0.75in]{geometry}
\usepackage{graphics,color,epstopdf,verbatim,amsmath,amsfonts,mathtools,amssymb,amsthm}
\usepackage[toc,page]{appendix}
\usepackage{xifthen} % provides \isempty test

\newcommand{\bq}{\begin{quotation}}
\newcommand{\eq}{\end{quotation}}
\newcommand{\be}{\begin{equation}}
\newcommand{\ee}{\end{equation}}
\newcommand{\bea}{\begin{eqnarray}}
\newcommand{\eea}{\end{eqnarray}}
\def\tr{{\rm tr}\,}

\newtheorem{theorem}{Theorem}
\newtheorem{lemma}{Lemma}

\DeclarePairedDelimiter\bra{\langle}{\rvert}
\DeclarePairedDelimiter\bbra{\langle\!\langle}{\|}

\DeclarePairedDelimiter\ket{\lvert}{\rangle}
\DeclarePairedDelimiter\kket{\|}{\rangle\!\rangle}

\DeclarePairedDelimiterX\braket[2]{\langle}{\rangle}{#1 \delimsize\vert #2}
\DeclarePairedDelimiterX\bbraket[2]{\langle\!\langle}{\rangle\!\rangle}{#1 \delimsize\| #2}
\DeclarePairedDelimiterX\cbraket[2]{(\!(}{)\!)}{#1 \delimsize\| #2}
\DeclarePairedDelimiterX\ketbra[2]{\lvert}{\rvert}{#1 \delimsize\rangle\!\langle #2}
\DeclarePairedDelimiterX\kketbra[2]{\|}{\|}{#1 \delimsize\rangle\!\rangle\!\langle\!\langle #2}
\DeclarePairedDelimiterX\cketbra[2]{\|}{\|}{#1 \delimsize)\!)\!(\!( #2}
\DeclarePairedDelimiterX\inner[2]{\langle}{\rangle}{#1,#2}
\DeclarePairedDelimiter\abs{\lvert}{\rvert}

\newcommand{\norm}[1]{\left\lVert#1\right\rVert}

\newcommand{\arxiv}[2][]{\ifthenelse{\isempty{#1}}{\href{http://arxiv.org/abs/#2}{{\tt arXiv:\allowbreak{}#2}}} {\href{http://arxiv.org/abs/#2}{{\tt arXiv:\allowbreak{}#2 [#1]}}}}
\newcommand{\booktitle}{\textsl}

\begin{document}
\title{Symmetric Informationally Complete Measurements Identify the Irreducible Difference between Classical and Quantum Systems}
\author{John B. DeBrota$^\dag$}
\author{Christopher A. Fuchs$^\dag$}
\author{Blake C. Stacey}
\affiliation{\medskip Department of Physics, University of Massachusetts Boston, 100 Morrissey Boulevard, Boston MA 02125, USA \medskip \\ $^\dag$Stellenbosch Institute for Advanced Study (STIAS), Wallenberg Research Center at Stellenbosch University, Marais Street, Stellenbosch 7600, South Africa\medskip}
\date{23 January 2020}
\begin{abstract}
We describe a general procedure for associating a minimal informationally complete quantum measurement (or MIC) with a purely probabilistic representation of the Born Rule. Such representations provide a way to understand the Born Rule as a consistency condition between probabilities assigned to the outcomes of one experiment in terms of the probabilities assigned to the outcomes of other experiments. In this
    setting, the difference between quantum and classical physics is the way their physical assumptions augment bare probability theory: Classical physics corresponds to a trivial augmentation---one just applies the Law of Total Probability (LTP) between the scenarios---while quantum theory makes use of the Born Rule expressed in one or another of the forms of our general procedure. To mark the \textit{irreducible} difference between quantum and classical, one should seek the representations that minimize the disparity between the expressions.  We prove that the representation of the Born Rule obtained from a \textit{symmetric} informationally complete measurement (or SIC) minimizes this distinction in at least two senses---the first to do with unitarily invariant distance measures between the rules, and the second to do with available volume in a reference probability simplex (roughly speaking a new kind of uncertainty principle).  Both of these arise from a useful result in majorization theory.  This work complements recent studies in quantum computation where the deviation of the Born Rule from the LTP is measured in terms of negativity of Wigner functions.
\end{abstract}

\maketitle

Quantum information theory represents a change of perspective. Rather than regarding quantum physics as a limitation on our abilities---the typical sentiment of older texts---we have learned that it can augment them. In frustrating some ambitions, it enables more subtle ones. Deviation from classicality is a \emph{resource,} and the idea that this resource can be quantified as a modification of the classical probability calculus dates to the beginning of the field~\cite{Feynman:1981}. More recent inquiries have developed this notion precisely:\ The ``negativity'' in a Wigner-func\-tion representation of quantum states is now understood to be valuable in its own right~\cite{Veitch:2014, Howard:2014, Delfosse:2015, Pashayan:2015, Stacey:2016, Zhu:2016, DeBrota:2017, Howard:2017, Raussendorf:2017, Kocia:2017}.  But what does this line of thinking say about quantum mechanics itself?  Can one, following the lead of Carnot, take what might seem a statement of ``mere'' engineering and find a physical principle?  In this paper, we prove some strong results in this regard in the context of finite dimensional Hilbert spaces.  In particular, we find the unique form of the quantum mechanical Born Rule that makes it resemble the classical Law of Total Probability (LTP) as closely as possible in at least two senses.  Both come from a significant majorization result which may be of general interest for resource theory.  This way of tackling the distinction between quantum and classical arises naturally in the quantum interpretive project of QBism~\cite{Fuchs:2013, Fuchs:2016}, where the Born Rule is seen as an empirically motivated constraint that one adds to probability theory when using it in the context of alternative (complementary) quantum experiments.  We expect the techniques developed here to give an alternative way to explore the paradigm of negativity and to be of use for a range of practical problems.

The standard procedure in quantum theory for generating probabilities starts with an observer, or agent, assigning a quantum state $\rho$ to a system. When the agent plans to measure the system, she represents the outcomes of her measurement with a positive operator-valued measure (POVM) $\{D_j\}$.  Assigning $\rho$ implies that she assigns the Born Rule probabilities $Q(D_j)=\tr \rho D_j$ for the outcomes of her measurement. In this way, any quantum state $\rho$ may be regarded as a compilation of probability distributions for all possible measurements. However, one does not have to consider all possible measurements to completely specify $\rho$.  In fact, there exist measurements which are informationally complete (IC) in the sense that $\rho$ is uniquely specified by the agent's expectations for the outcomes of that single measurement~\cite{Prugovecki:1977}. With respect to an IC measurement, any quantum state, pure or mixed, is equivalent to a single probability distribution. In this paper, we consider minimal informationally-complete POVMs (MICs) for finite dimensional quantum systems.  These sets of operators form bases for the vector space of Hermitian operators and lead to probability distributions with the fewest number of entries necessary for reconstructing the quantum state. MICs furnish a convenient way to bypass the language of quantum states, making quantum theory analogous to classical stochastic process theory, in which one puts probabilities in and gets probabilities out.

One can eliminate the need to use the operators $\rho$ and $D_j$ in the Born Rule by reexpressing it as a relation between an agent's expectations for different experiments. Suppose our agent has a preferred reference process consisting of a measurement to which she ascribes the MIC $\{H_i\}$, and, upon obtaining outcome $i$, the preparation of a state $\sigma_i$, drawn from a linearly independent set of post-measurement states $\{\sigma_i\}$. (See Fig.\ \ref{figure1}.) In her choice of this reference process, she requires
linearly independent post-measurement states so that the inner products $\tr D_j\sigma_i$ will uniquely characterize the operators $D_j$. Let $P(H_i)$ be her probabilities for the measurement $\{H_i\}$ and $P(D_j|H_i)$ be her conditional probabilities for a subsequent measurement of $\{D_j\}$. What consistency requirement among $Q(D_j)$, $P(H_i)$, and $P(D_j|H_i)$ does quantum physics entail?

Using the fact that $\{\sigma_i\}$ is a basis, we may write
\begin{equation}
    \rho=\sum_j\alpha_j\sigma_j\;,
\end{equation}
for some set of real coefficients $\alpha_j$.  The probability of outcome $H_i$ is then
\begin{equation}
    P(H_i)=\sum_j\alpha_j\,\tr H_i\sigma_j=\sum_j\,\big[\Phi^{-1}\big]_{ij}\,\alpha_j\;,
\end{equation}
where we have defined the matrix $\Phi$ via its inverse,
\begin{equation}\label{phiinv}
    \big[\Phi^{-1}\big]_{ij}:=\tr H_i\sigma_j=h_i\tr\rho_i\sigma_j\;,
\end{equation}
for $\rho_i:=H_i/h_i$ and $h_i:=\tr H_i$. The invertibility of $\Phi$ is assured by the linear independence of the MIC and post-measurement sets. This implies that the coefficients of $\rho$ in the $\sigma_i$ basis may be written as an application of the $\Phi$ matrix on the vector of probabilities,
\begin{equation}\label{rhoinbasis}
    \rho=\sum_i\left[\sum_k[\Phi]_{ik}P(H_k)\right]\! \sigma_i\;.
\end{equation}
The probability of $D_j$ is given by another application of the Born Rule, which becomes
\begin{equation}\label{ltpanalogindices}
                Q(D_j)=\sum_{i=1}^{d^2}\left[\sum_{k=1}^{d^2}[\Phi]_{ik}P(H_k)\right]\! P(D_j|H_i)\;,
\end{equation}
where $P(D_j|H_i)=\tr D_j\sigma_i$ is the probability for outcome $D_j$ conditioned on obtaining $H_i$ in the reference measurement.
In more compact matrix notation, we can write
\begin{equation}\label{ltpanalog}
    Q(D)=P(D|H)\,\Phi\, P(H)\;,
\end{equation}
where $P(D|H)$ is a matrix of conditional probabilities.
\begin{figure}
\includegraphics[width=\linewidth]{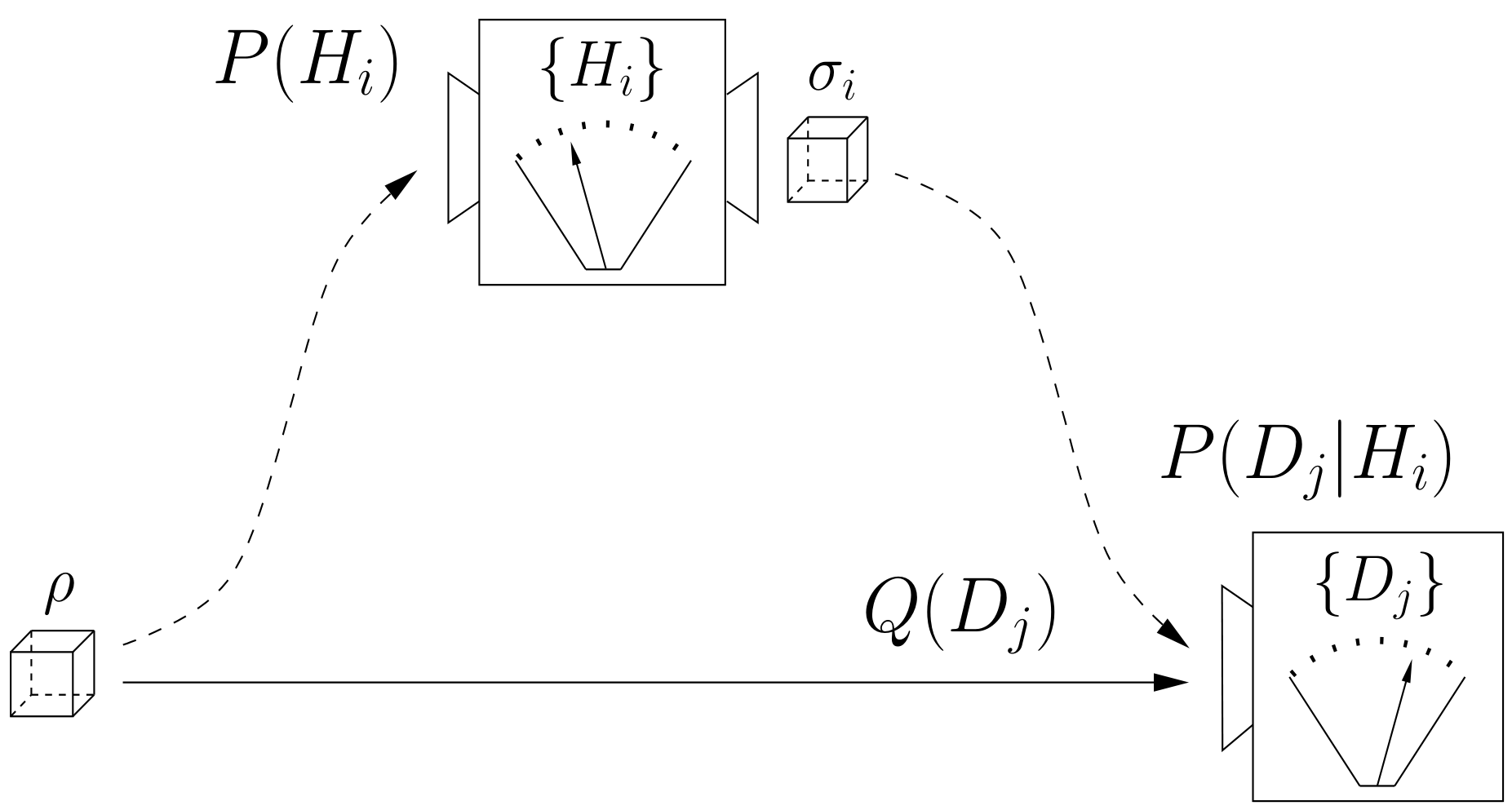}
    \caption{The solid and dashed lines represent two hypothetical procedures an agent contemplates for a system assigned state $\rho$. The solid line represents making a direct measurement of a POVM $\{D_j\}$. The dotted line represents making the MIC measurement $\{H_i\}$ first, preparing a post-measurement state $\sigma_i$, and then finally making the $\{D_j\}$ measurement. For the solid path, the agent assigns one set of probabilities $Q(D_j)$. For the dotted path, she assigns two sets of
    probabilities: $P(H_i)$ and $P(D_j|H_i)$. Unadorned by physical assumptions, probability theory does not suggest a relation between these paths. The Born Rule in the form of Eq.~\eqref{ltpanalog} is such a relation.}
\label{figure1}
\end{figure}

A SIC~\cite{Zauner:1999, Renes:2004, Scott:2010, Appleby:2014, Appleby:2016, Appleby:2017, Scott:2017, Grassl:2017, Fuchs:2017a, Stacey:2017} is a MIC for which all the $H_i$ are rank-1 and
\begin{equation}
  \tr H_i H_j = \frac{1}{d^2} \frac{d\delta_{ij} + 1}{d+1}\;.
\end{equation}
SICs have yet to be proven to exist in all finite dimensions $d$, but they are widely believed to \cite{Fuchs:2017a} and have even been experimentally demonstrated in some low dimensions~\cite{Durt:2008,Medendorp:2011,Bent:2015,Hou:2018}.  The \emph{SIC projectors\/} associated with a SIC are the pure states $\rho_i=dH_i$. In dimension 2, a SIC can be represented as a regular tetrahedron inscribed in the Bloch sphere.  (States defining a qubit SIC can be extracted from Feynman's
1987 essay ``Negative probabilities''~\cite{Feynman:1987}.) In higher dimensions, they are, of course, harder to visualize. When there is no chance of confusion, we will refer to the set of projectors as SICs as well. Prior work has given special attention to the reference
procedure where the measurement and post-measurement states are the same SIC~\cite{Fuchs:2013,Fuchs:2017b,Appleby:2017b}.  In this case we denote $\Phi$ by $\Phi_{\rm SIC}$ and Eq.~\eqref{ltpanalogindices} takes the particularly simple form
\begin{equation}
    Q(D_j)=\sum_{i=1}^{d^2}\left[(d+1)P(H_i)-\frac{1}{d}\right]\! P(D_j|H_i)\;.
    \label{urgleichung}
\end{equation}

Recall that the LTP expresses the simple consistency relation between the probabilities one assigns to the second of a sequence of measurements, the probabilities one assigns to the first, and the conditional probabilities for the second given the outcome of the first. Written in vector notation, this is
\begin{equation}\label{ltp}
    P(D)=P(D|H)P(H)\;.
\end{equation}
We write $P(D)$ as opposed to $Q(D)$ to indicate that it is the probability vector for the \textit{second} of two measurements. $Q(D)$, on the other hand, is the vector of probabilities associated with a single measurement. Aside from the presence of $\Phi$ matrix, Eq.~\eqref{ltpanalog} is functionally equivalent to the LTP.

Although $P(H)$, $P(D|H)$, and $Q(D)$ are probabilities, $\Phi P(H)$
often is not. One may see by summing both sides of Eq.~\eqref{ltpanalogindices} over $j$ that the vector is normalized, but
in general it may contain negative numbers and values greater than
$1$. Such a vector is known as a quasiprobability, and matrices like
$\Phi$---real-valued matrices with columns summing to $1$---which take probabilities to quasiprobabilites are called
column-quasistochastic matrices~\cite{VanDeWetering:2018}.  The subset of column-quasistochastic matrices with nonnegative entries are the column-stochastic matrices. The inverse of a
column-stochastic matrix is generally a column-quasistochastic matrix; in
our case, inspection of Eq.~\eqref{phiinv} reveals that
$\Phi^{-1}$ is column-stochastic.

What would it mean if $\Phi$ could equal $I$? In this case we would have $Q(D)=P(D)$.
Then, conceptually, it wouldn't matter if the intermediate measurement were performed or not. Put another way, we could behave as though measurements simply revealed a preexisting property of the system, as in classical physics where measurements provide information about a system's coordinates in phase space.

Some amount of what makes quantum theory nonclassical resides in the fact that $\Phi$ cannot equal $I$. How close, then, can we make $\Phi$ to $I$ by wisely choosing our MIC and post-measurement states? It turns out that $\Phi_{\rm SIC}$ is closest to the identity with respect to the distance measure induced by any member of a large family of operator norms called unitarily invariant norms~(see section 3.5 in \cite{Horn:1994}).  A unitarily invariant norm is one such that $\norm{A}=\norm{UAV}$
for all unitary matrices $U$ and $V$. These norms include the Schatten $p$-norms (among which are the trace norm, the Frobenius norm, and the operator norm when $p=1,2,$ and $\infty$ respectively) and the Ky Fan $k$-norms. This result codifies the intuition that Eq.~\eqref{urgleichung} represents the ``simplest modification one can imagine to the LTP'' \cite[p.\ 1971]{Fuchs:2014}.

To prove this, we will make use of the theory of majorization \cite{Horn:1994,Marshall:2011}. Suppose $x$ and $y$ are vectors of $N$ real numbers and that $x^\downarrow$ and $y^\downarrow$ are $x$ and $y$ sorted in nonincreasing order. Then we say that $x$ weakly majorizes $y$ from below, denoted $x\succ_w y$, if
\begin{equation}\label{MajBaker}
    \sum_{i=1}^k x_i^\downarrow\geq\sum_{i=1}^k y_i^\downarrow\;,\quad \text{for } k=1,\ldots,N\;.
\end{equation}
If the last inequality is an equality, we say $x$ majorizes $y$, denoted $x\succ y$.

Another variant of majorization, called log majorization or multiplicative majorization, is also studied \cite{Marshall:2011}. We say that $x$ weakly log majorizes $y$ from below, denoted $x\succ_{w\log}y$, if
\begin{equation}\label{logmaj}
    \prod_{i=1}^k x_i^\downarrow\geq\prod_{i=1}^ky_i^\downarrow\;, \quad \text{for } k=1,\ldots,N\;.
\end{equation}
If the last inequality is an equality, we say $x$ log majorizes $y$, denoted $x\succ_{\rm log} y$.
Taking the log of both sides of Eq.~\eqref{logmaj} demonstrates that log majorization is majorization between the vectors after an element-wise application of the log map. Log majorization is strictly stronger than regular majorization; $x\succ_{w\log}y\implies x\succ_wy$, but the reverse implication is not true. Majorization is a partial order on vectors of real numbers sorted in nonincreasing order.

Throughout this paper we will make use of the standard inequalities between the arithmetic, geometric, and harmonic means for vectors of $n$ positive numbers $x_i$:
\begin{equation}\label{meanordering}
    \frac{1}{n}\sum_{i=1}^nx_i\ge\left(\prod_{i=1}^nx_i\right)^{\!1/n}\ge\left(\frac{1}{n}\sum_{i=1}^n\frac{1}{x_i}\right)^{\!-1}\;.
\end{equation}
with equality in all cases if and only if $x_i=c$ for all $i$.
We now turn to two lemmas.
\begin{lemma}\label{det}
    Let $\Phi_{\rm p}$ denote the column-quasistochastic matrix associated with a MIC and a proportional post-measurement set. Then $\det\Phi_{\rm p}\geq\det\Phi_{\rm SIC}$ with equality iff the MIC is a SIC.
\end{lemma}
\begin{proof}
We may write $\Phi^{-1}_{\rm p}=GA^{-1}$ where $G_{ij}:=\tr H_iH_j$ is the Gram matrix of the MIC elements and $A_{ij}:=h_i\delta_{ij}$. Note that $\Phi^{-1}_{\rm p}$ has real, positive eigenvalues because it has the same spectrum as the positive definite matrix $A^{-1/2}GA^{-1/2}$. Also note that
    \begin{equation}
        \sum_i\frac{1}{\lambda_i(\Phi_{\rm p})}=\tr\Phi_{\rm p}^{-1}=\sum_ih_i\tr\rho_i\sigma_i\leq\sum_ih_i=d\;.
    \end{equation}
One of the eigenvalues of $\Phi_{\rm p}$, which we denote $\lambda_{d^2}(\Phi_{\rm p})$, must equal $1$ because an equal-entry row vector is always a left eigenvector with eigenvalue $1$ of a matrix with columns summing to unity. Therefore, we may write
    \begin{equation}\label{reciprocalbound}
    \sum_{i<d^2}\frac{1}{\lambda_i(\Phi_{\rm p})}\leq d-1.
\end{equation}
The reciprocal of this expression is proportional to the harmonic mean of the first $d^2-1$ eigenvalues of $\Phi_{\rm p}$. Thus, because the geometric mean is always greater than or equal to the harmonic mean,
    \begin{equation}
    \left(\prod^{d^2-1}_{i=1}\lambda_i(\Phi_{\rm p})\right)^{\!\frac{1}{d^2-1}}\geq\left(\frac{1}{d^2-1}\sum_{i=1}^{d^2-1}\frac{1}{\lambda_i(\Phi_{\rm p})}\right)^{\!-1}\geq d+1\;,
\end{equation}
which, noting that $\lambda_{d^2}(\Phi_{\rm p})=1$, implies
\begin{equation}\label{mindet}
    \det\Phi_{\rm p}\geq(d+1)^{d^2-1}=\det\Phi_\text{SIC}\;.
\end{equation}
Equality is achieved in this iff all the $\lambda_i(\Phi_{\rm p})$ are equal, so Eq.~\eqref{mindet} is saturated iff $\lambda(\Phi_{\rm p})=\lambda(\Phi_{\rm SIC})$. We next show this implies that in fact the MIC is a SIC.

For any $\Phi_{\rm p}^{-1}$, we may write $\Phi_{\rm p}^{-1}=P^{-1}DP$ where the rows of $P$ are the left-eigenvectors of $\Phi_{\rm p}^{-1}$ and $D$ is the diagonal matrix of eigenvalues of $\Phi_{\rm p}^{-1}$. Since $\Phi_{\rm p}^{-1}$ is column-stochastic, the row vector $(1/d,\dots,1/d)$ is the (scaled) left-eigenvector of $\Phi_{\rm p}^{-1}$ with eigenvalue $1$, and so it is the first row of $P$ when the eigenvalues are in descending order. Left-eigenvectors of a matrix are right-eigenvectors of the transpose of the matrix, so we have
    \begin{equation}
        \begin{split}
        (\Phi_{\rm p}^{-1})^T\ket{v}&=A^{-1}G\ket{v}=A^{-1}GA^{-1}A\ket{v}\\
        &=A^{-1}\Phi_{\rm p}^{-1}A\ket{v}=\lambda\ket{v}\;,\\
        \implies\Phi_{\rm p}^{-1}&A\ket{v}=\lambda A\ket{v}\;,
        \end{split}
    \end{equation}
    where $\bra{v}$ is an arbitrary left-eigenvector of $\Phi_{\rm p}^{-1}$. Combined with our choice of scale for the first row of $P$, we conclude that the first column of $P^{-1}$ is $(h_1,h_2,\dots,h_{d^2})^T$.

    Now suppose $\Phi_{\rm p}$ is such that $\lambda(\Phi_{\rm p})=\lambda(\Phi_{\rm SIC})$. Then $G=P^{-1}DPA$ where $[D]_{ij}=\frac{1}{d+1}(\delta_{ij}+d\delta_{i1}\delta_{j1})$, and
    \begin{eqnarray}\label{gram}
        %\begin{split}
            [G]_{ij}&=&\sum_{klm}[P^{-1}]_{ik}[D]_{kl}[P]_{lm}[A]_{mj}\nonumber\\
            &=&\sum_{klm}[P^{-1}]_{ik}\left[\frac{1}{d+1}(\delta_{kl}+d\delta_{k1}\delta_{l1})\right][P]_{lm}\delta_{mj}h_m\nonumber\\
            %&=\sum_{kl}[P^{-1}]_{ik}\left[\frac{1}{d+1}(\delta_{kl}+d\delta_{k1}\delta_{l1})\right][P]_{lj}h_j\nonumber\\
            &=&\frac{1}{d+1}\sum_{kl}[P^{-1}]_{ik}(\delta_{kl}+d\delta_{k1}\delta_{l1})[P]_{lj}h_j\nonumber\\
            &=&\frac{1}{d+1}(h_j\delta_{ij}+dh_j[P^{-1}]_{i1}[P]_{1j})\nonumber\\
            &=&\frac{1}{d+1}(h_j\delta_{ij}+h_ih_j)\;.
    %\end{split}
    \end{eqnarray}
    In the last step we used that $[P]_{1j}=1/d$ and $[P^{-1}]_{i1}=h_i$. If this Gram matrix comes from a MIC, one may use
    \begin{equation}
        [G]_{ii}=h_i^2\tr\rho_i^2=\frac{1}{d+1}(h_i+h_i^2)\;,
    \end{equation}
    and the fact that $\tr\rho_i\leq1$ to show that $h_i\geq 1/d$. As the average $h_i$ value must be $1/d$, this implies that $h_i=1/d$ for all $i$ and furthermore that each $\rho_i$ is rank-$1$. Substituting this into Eq.~\eqref{gram} gives
    \begin{equation}
        [G]_{ij}=\frac{d\delta_{ij}+1}{d^2(d+1)}\;,
    \end{equation}
    that is, the MIC is a SIC and $\Phi_{\rm p}=\Phi_\text{SIC}$.
\end{proof}
Let $s(A)$ denote the vector of singular values of the matrix $A$ in nonincreasing order. The proof of the following lemma may be found in Appendix~A.
\begin{lemma}\label{maj}
  Let $\Phi$ be the column-quasistochastic matrix associated with an arbitrary reference process. Then
  \begin{equation}
    |\det{\Phi}| \geq \det{\Phi_{\rm SIC}} \;,
  \end{equation}
  and
  \begin{equation}
    s(\Phi)\succ_{w\log}s(\Phi_{\rm SIC})\;,
  \end{equation}
  with equality iff the MIC and post-measurement states are SICs. In that case,
  \begin{equation}
    \Phi^\dag \Phi = \Phi \Phi^\dag = \Phi_{\rm SIC}^2\;.
  \end{equation}
\end{lemma}

We are now poised to prove:
\begin{theorem}\label{unitarilyinvariant}
    Let $\Phi$ be the column-quasistochastic matrix associated with an arbitrary reference process. Then for any unitarily invariant norm $\norm{\cdot}$,
    \begin{equation}
        \norm{I-\Phi}\geq\norm{I-\Phi_{\rm SIC}}\;,
    \end{equation}
    with equality iff the MIC and post-measurement states are the same SIC.
\end{theorem}
\begin{proof}
    By Corollary 3.5.9 in \cite{Horn:1994}, every unitarily invariant norm is monotone with respect to the partial order on matrices induced by weak majorization of the vector of singular values. $I-\Phi$ is singular with exactly one eigenvalue equal to zero, so one
of its singular values is zero as well. Then
\begin{equation}
    \begin{split}
    s(I-\Phi)&\succ\left\{\frac{\sum_is_i(I-\Phi)}{d^2-1},\ldots,\frac{\sum_is_i(I-\Phi)}{d^2-1}\right\}\\
    &\succ_w\{d,\ldots,d\}=s(I-\Phi_{\rm SIC})
\end{split}
\end{equation}
if
\begin{equation}
    \sum_is_i(I-\Phi)\geq d(d^2-1)\;.
\end{equation}
We have
\begin{eqnarray}\label{singsum}
    \sum_is_i(I-\Phi)&\geq&\sum_i|\lambda_i(I-\Phi)|=\sum_i|\lambda_i(\Phi)-1|\nonumber\\
    &\geq&\sum_i(|\lambda_i(\Phi)|-1)\geq\sum_i\lambda_i(\Phi_{\rm SIC})-d^2\nonumber\\
    &=&d(d^2-1)\;,
\end{eqnarray}
where the first inequality follows Eq.~3.3.13a in \cite{Horn:1994}, the second follows from the triangle inequality, and the last follows from Lemma \ref{maj}.

It remains to show that the inequality is saturated only if the MIC and post-measurement states are the same SIC. From Lemma 2, we know that $|\det{\Phi}| \geq \det{\Phi_{\rm SIC}}$ with equality iff both the measurement and post-measurement states are SICs. Then, using the mean orderings \eqref{meanordering}, we can conclude the inequality
\begin{equation}
\begin{split}
    \sum_{i=1}^{d^2}|\lambda_i(\Phi)|&=1+\sum_{i=1}^{d^2-1}|\lambda_i(\Phi)|\\
    &\geq 1+(d^2-1)\left(\prod_{i=1}^{d^2-1}|\lambda_i(\Phi)|\right)^\frac{1}{d^2-1}\\
    &=1+(d^2-1)\abs{\det{\Phi}}^\frac{1}{d^2-1}\\
    &\geq 1+(d^2-1)|\det{\Phi_{\rm SIC}}|^\frac{1}{d^2-1}=\sum_{i=1}^{d^2}\lambda_i(\Phi_{\rm SIC})
\end{split}
\end{equation}
is also saturated iff both the measurement and post-measurement states are SICs which, in turn, implies the same for the last inequality of \eqref{singsum}.

If both the measurement and the post-measurement states are SICs, then Lemma 2 implies that $\Phi$ commutes with its own adjoint, making it a normal matrix. In turn, this implies that $|\lambda_i(\Phi)| = s_i(\Phi) = s_i(\Phi_{\rm SIC})$. The spectrum of $\Phi$ is the spectrum of $\Phi_{\rm SIC}$, but with nonzero phases allowed: $((d+1)e^{i\theta_1},\ldots,(d+1)e^{i\theta_{d^2-1}}, 1)$. Subject to this constraint, the triangle inequality is saturated if and only if all the eigenvalues are real and positive, which is also when $\tr\Phi$ and $\tr\Phi^{-1}$ are maximized. Considering the latter, the diagonal elements are $(1/d)\tr\Pi_i\Pi'_i$ for SIC projectors $\{\Pi_i\}$ and $\{\Pi'_i\}$ and so the maximum trace is obviously attained only if $\Pi_i=\Pi'_i$ for all $i$. So, in order for the sum of the singular values $s_i(I-\Phi)$ to attain its lower bound, both sets must be the same SIC.
\end{proof}

It is known that no quasiprobability representation of
quantum theory can be entirely nonnegative~\cite{Ferrie:2009}. What
does this mean in our formalism?
\begin{figure}
  \includegraphics[width=\linewidth]{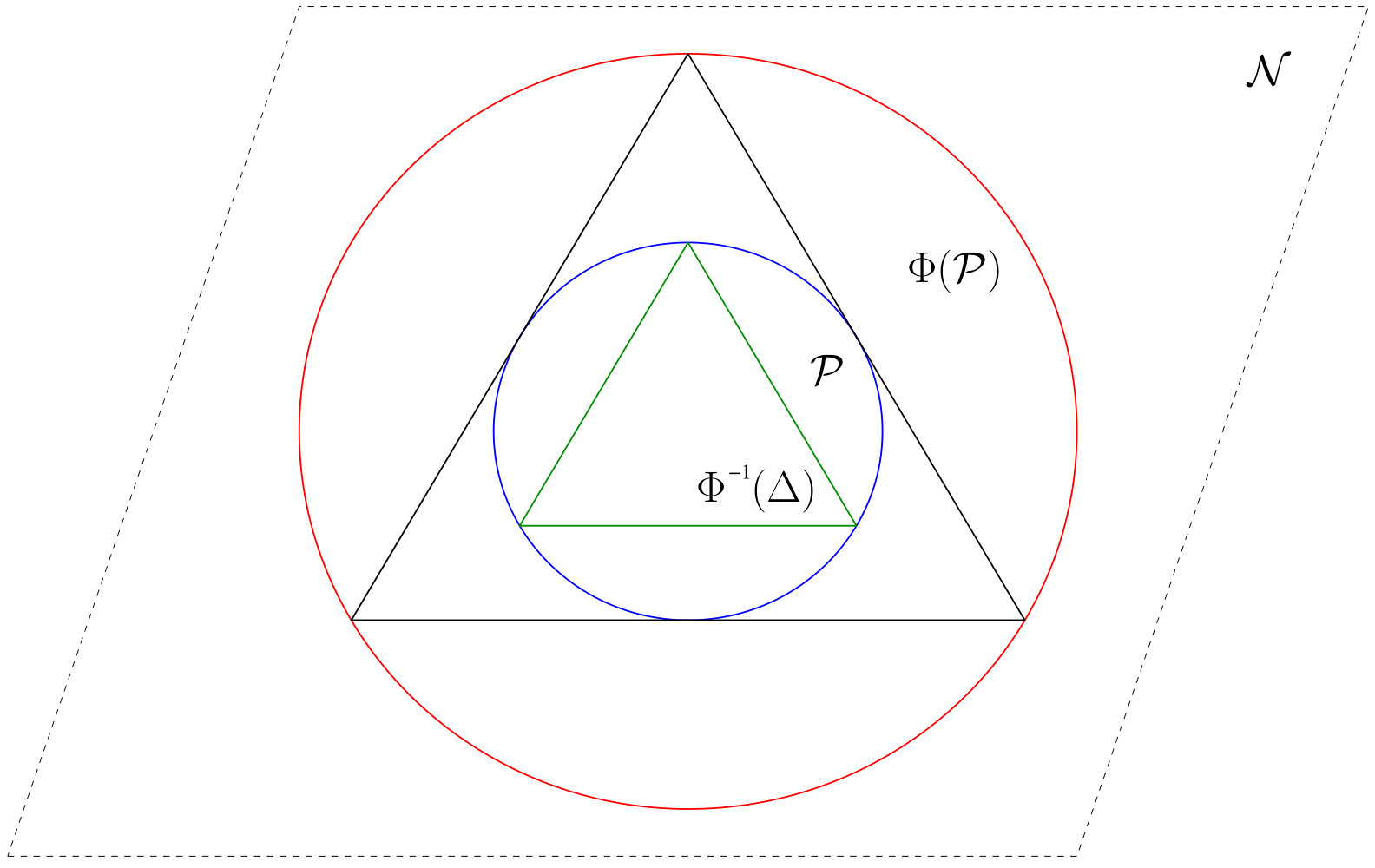}
\caption{$\mathcal{N}$ is the normalized hyperplane of $d^2$-element quasiprobability vectors and the outer, black triangle represents the $(d^2-1)$-simplex $\Delta$ of probabilities. For a given MIC, the inner, green triangle is the simplex $\Phi^{-1}(\Delta)$, the blue circle is the image of $\mathcal{Q}_d$ under the Born Rule, denoted $\mathcal{P}$, and the red circle is $\Phi(\mathcal{P})$. $\mathcal{P}$ and $\Phi(\mathcal{P})$ are portrayed with
circles to capture convexity and inclusion relationships only; they need not bear any resemblance to spheres.}
\label{figure2}
\end{figure}

Let $\mathcal{N}$ be the normalized hyperplane of $d^2$-element quasiprobability vectors. Within this is the $(d^2-1)$-simplex of probability vectors, $\Delta$\@. For any MIC, $d$-dimensional quantum state space $\mathcal{Q}_d$ is mapped by the Born Rule to a convex subset of $\Delta$, denoted $\mathcal{P}$. Note that $\Phi^{-1}(\Delta)$ is equal to the convex hull of the $d^2$ probability vectors $\tr H_j\sigma_i$, that is, the probabilities for the MIC measurement for each post-measurement state.
Consequently, $\Phi^{-1}(\Delta)\subset \mathcal{P}$, which implies $\Delta\subset \Phi(\mathcal{P})$. These inclusions must be strict, i.e., $\Phi\neq I$: When the MIC and post-measurement states are rank-$1$, the vertices of the simplex will be among the pure-state probability vectors, but $\mathcal{P}$ contains more pure states than there are vertices of $\Phi^{-1}(\Delta)$. Since the image of some probability vectors consistent with quantum theory must leave the probability simplex under the application of $\Phi$, we have demonstrated that the appearance of negativity is unavoidable in our framework and is in fact characterized by the fact that $\Phi$ cannot equal the identity. Figure \ref{figure2} illustrates the situation.

The weak log majorization result of Lemma \ref{maj} has at least one more important implication for quantifying the quantum deviation from classicality. Instead of looking at the functional form of Eq.~\eqref{ltpanalog} and considering how much of a deviation from the LTP it represents, one may approach the problem from a geometric perspective.

Classically one can always imagine assigning probability 1 to an outcome of a putative ``maximally informative measurement''---for instance when one knows the system's exact phase space point. However, in an interpretation of quantum theory without hidden variables, whatever one might mean by ``maximally informative,'' one cannot mean that the reference measurement's full probability simplex is available. Indeed, quantum mechanics does not allow probability 1 for the outcome of any MIC
measurement~\cite{Fuchs:2002}. Thus deviation from classicality can also be captured by the fact that the region of probabilities compatible with quantum states is strictly smaller than the full ($d^2-1$)-simplex. In this setting, the irreducible deviation from classicality is defined by the largest possible region for a reference measurement's probability simplex. The following theorem establishes that a SIC measurement uniquely maximizes the Euclidean volume of this region, thereby
answering a question raised by one of us in 2002 \cite[pp.\ 475, 571]{Fuchs:2014}.

\begin{theorem}\label{vol}
    For any MIC in dimension $d$, let $\mathcal{P}$ denote the image of $\mathcal{Q}_d$ under the Born Rule and let $\textnormal{vol}_\textnormal{E}(\mathcal{P})$ denote its Euclidean volume. Then
    \begin{equation}
        \textnormal{vol}_\textnormal{E}(\mathcal{P})\leq\textnormal{vol}_\textnormal{E}(\mathcal{P_\textnormal{SIC}})\;,
\end{equation}
with equality iff the MIC is a SIC. Furthermore,
\begin{equation}\label{volPsic}
    \textnormal{vol}_\textnormal{E}(\mathcal{P_\textnormal{SIC}})=\sqrt{\frac{(2\pi)^{d(d-1)}}{d^{d^2-2}(d+1)^{d^2-1}}}\frac{\Gamma(1)\cdots\Gamma(d)}{\Gamma(d^2)}\;.
\end{equation}
\end{theorem}
\noindent The proof of Theorem \ref{vol} involves methods of differential geometry which would be distracting here. We direct the interested reader to Appendix~B for details.

The $(d^2-1)$-simplex $\Delta$ has Euclidean volume~\cite{Caves:2001}
\begin{equation}\label{volsimplex}
      \text{vol}_\text{E}(\Delta)=\frac{d}{\Gamma(d^2)}\;,
  \end{equation}
  so we can calculate the ratio of the Euclidean volumes of $\mathcal{P}_\text{SIC}$ and the simplex it lies within,
  \begin{equation}
      \frac{\textnormal{vol}_E(\mathcal{P}_\text{SIC})}{\text{vol}_\text{E}(\Delta)}=\sqrt{\frac{(2\pi)^{d(d-1)}}{d^{d^2}(d+1)^{d^2-1}}}\Gamma(1)\cdots\Gamma(d)\;.
  \end{equation}

  When $d=2$, quantum state space is the Bloch ball and $\mathcal{P}_\text{SIC}$ is the largest ball which can be inscribed in the regular tetrahedron $\Delta_3$,
  \begin{equation}
      \frac{\text{vol}_\text{E}(\mathcal{P}_\text{SIC})}{\text{vol}_\text{E}(\Delta_3)}=\frac{\pi}{6\sqrt{3}}\approx 0.3023\;.
  \end{equation}
When $d=3$,
  \begin{equation}
      \frac{\text{vol}_\text{E}(\mathcal{P}_\text{SIC})}{\text{vol}_\text{E}(\Delta_8)}=\frac{\pi^3}{1296\sqrt{3}}\approx 0.0138\;.
  \end{equation}
In general, the ratio is very rapidly decreasing, signifying a greater and greater deviation from classicality with each Hilbert space dimension.

Theorems \ref{unitarilyinvariant} and \ref{vol} show that the SICs provide a way of casting the Born Rule in wholly probabilistic terms, which by two different standards make the difference between classical and quantum as small as possible.  Of all the representations deriving from our general procedure, the representation given by Eq.~\eqref{urgleichung} is the essential one for specifying how quantum is quantum.\medskip

{\bf Acknowledgements}.  CAF thanks K.~Kato and M.~P. M\"uller for discussions and A.~Khrennikov for accusing us of being ``really addicted'' to the SICs for no good cause~\cite{Khrennikov:2017}---this paper attempts to show him he was wrong.  JBD and CAF were supported in part by the Foundational Questions Institute Fund on the Physics of the Observer (grant FQXi-RFP-1612), a donor advised fund at the Silicon Valley Community Foundation.

%%%START old appendix A
\begin{comment}
\appendix
\label{appendix:equal-spectra}
\section{Appendix A}

\end{comment}

%%%END old appendix A

\appendix
\label{appendix:equal-spectra}
\section{Appendix A: Proof of Lemma \ref{maj}}
For a MIC $\{E_i\}$ and a post-measurement set $\{\sigma_j\}$,
\begin{equation}
    \left[\Phi^{-1}\right]_{ij}=\tr E_i\sigma_j\;.
\end{equation}
The elements of the MIC may be expanded in the SIC basis
\begin{equation}\label{alpha}
    E_i=\sum_k\left[\alpha\right]_{ik}H_k\;,
\end{equation}
so we may write
\begin{equation}
    \left[\Phi^{-1}\right]_{ij}=\sum_k\left[\alpha\right]_{ik}\tr H_k\sigma_j=\sum_k\left[\alpha\right]_{ik}p(k|j)\;,
\end{equation}
where $p(k|j)$ is the probabilistic representation of the state $\sigma_j$ with respect to the SIC $\{H_k\}$. The $\alpha$ matrix must be invertible because it is a transformation between two bases, so the probability vectors can be written
\begin{equation}\label{alphaconditionals}
    p(i|j)=\sum_k\left[\alpha^{-1}\right]_{ik}\left[\Phi^{-1}\right]_{kj}\;.
\end{equation}
We know that SIC probability vectors satisfy \cite{Fuchs:2013}
\begin{equation}
    \sum_ip(i|j)^2\leq\frac{2}{d(d+1)}\quad \forall j\;,
\end{equation}
so we have
\begin{equation}\label{sicbound}
    \sum_{i}\left(\sum_k\left[\alpha^{-1}\right]_{ik}\left[\Phi^{-1}\right]_{kj}\right)^{\!2}\leq \frac{2}{d(d+1)}\quad \forall j\;.
\end{equation}
Summing over $j$, we then have
\begin{equation}
    \sum_{ij}\left(\sum_k\left[\alpha^{-1}\right]_{ik}\left[\Phi^{-1}\right]_{kj}\right)^{\!2}\leq \frac{2d}{d+1}\;.
\end{equation}
This expression is the sum of the absolute square entries of a matrix, which is equivalent to the square of the Frobenius norm of the matrix:
\begin{equation}
    \norm{\alpha^{-1}\Phi^{-1}}_2^2=\sum_is^2(\alpha^{-1}\Phi^{-1})\leq\frac{2d}{d+1}\;.
\end{equation}
From \cite{Horn:1994} 3.1.11, for any square matrix $A$,
\begin{equation}
    \sum_i|\lambda_i(A)|^2\leq\sum_i\varsigma^2(A)\;,
\end{equation}
so we have a general bound on the absolute squared spectrum:
\begin{equation}\label{SICupperbound}
    \sum_i|\lambda_i(\alpha^{-1}\Phi^{-1})|^2\leq \frac{2d}{d+1}\;.
\end{equation}
Eq.~\eqref{alphaconditionals} shows that $\alpha^{-1}\Phi^{-1}$ is column-stochastic and thus that one of its eigenvalues is $1$, so we may write: 
\begin{equation}
    \sum_{i>1}|\lambda_i(\alpha^{-1}\Phi^{-1})|^2\leq\frac{2d}{d+1}-1=\frac{d-1}{d+1}\;.
\end{equation}
Now, using the arithmetic-geometric mean inequality,
\begin{equation}
    \begin{split}
    \frac{d-1}{d+1}&\geq\sum_{i>1}|\lambda_i(\alpha^{-1}\Phi^{-1})|^2\\
    &\geq(d^2-1)\left(\prod_{i>1}|\lambda_i(\alpha^{-1}\Phi^{-1})|^2\right)^{\frac{1}{d^2-1}}\\
    &=(d^2-1)|\det\alpha^{-1}\Phi^{-1}|^{\frac{2}{d^2-1}}\;,
\end{split}
\end{equation}
which implies
\begin{equation}
    \begin{split}
    |\det\alpha^{-1}\Phi^{-1}|&\leq\left(\frac{d-1}{(d+1)(d^2-1)}\right)^{\frac{d^2-1}{2}}\\
    &=\left(\frac{1}{d+1}\right)^{d^2-1}=\det\Phi^{-1}_\text{SIC}\;.
\end{split}
\end{equation}
From Eq.~\eqref{alpha}, we can write
\begin{equation}
    \begin{split}
\tr E_iE_j&=\sum_{kl}\alpha_{ik}\alpha_{jl}\tr H_kH_l\iff G=\alpha G_{\rm SIC}\alpha^T\\
&\iff \det G=(\det \alpha)^2\det G_{\rm SIC}\;,
\end{split}
\end{equation}
where $G$ is the MIC Gram matrix and $G_{\rm SIC}$ is the SIC Gram matrix. Recall the definition of the $A$ matrix from the proof of Lemma \ref{det}. The arithmetic-geometric mean inequality shows $\det A\leq(1/d)^{d^2}$ with equality iff $h_i=1/d$. Then, since $G=\Phi_{\rm p}^{-1}A$, Lemma \ref{det} shows
\begin{eqnarray}\label{maxdet}
    \det G =(\det\Phi_{\rm p}^{-1})(\det A)&\leq& (\det\Phi^{-1}_{\rm SIC})(1/d)^{d^2}\nonumber\\
    &=&\det G_{\rm SIC}\;,
\end{eqnarray}
with equality iff the MIC is a SIC. This implies $(\det\alpha)^2\leq 1$, and so $|\det\alpha|\leq 1$. Since $|\det\alpha^{-1}\Phi^{-1}|=|\det\alpha^{-1}||\det\Phi^{-1}|$, we conclude that
\begin{equation}
    |\det\Phi^{-1}|\leq\det\Phi^{-1}_\text{SIC}\;.
\end{equation}
Equivalently, $\det\Phi_\text{SIC}\leq |\det\Phi|$. Theorem 3.3.2 in \cite{Horn:1994} shows $s(A)\succ_{\log} |\lambda(A)|$ for an arbitrary matrix $A$. To show the desired weak log majorization result, we wish to prove $|\lambda(\Phi)|\succ_{w\log}\lambda(\Phi_\text{SIC})$. For this we show weak majorization of the log of the entries. 
\begin{equation}
        \begin{split}
            \log|\lambda(\Phi)|&\succ\left(\frac{\sum_{i=1}^{d^2}\log|\lambda_i(\Phi)|}{d^2-1},\dots,\frac{\sum_{i=1}^{d^2}\log|\lambda_i(\Phi)|}{d^2-1},0\right)\\
            &=\left(\frac{\log|\det\Phi|}{d^2-1},\ldots,\frac{\log|\det\Phi|}{d^2-1},0\right)\\
            &\succ_w\left(\frac{\log\det\Phi_\text{SIC}}{d^2-1},\dots,\frac{\log\det\Phi_\text{SIC}}{d^2-1},0\right)\\
            &=(\log(d+1),\ldots,\log(d+1),0)=\lambda(\log\Phi_\text{SIC})\;.
    \end{split}
    \end{equation}
    Thus,
\begin{equation}
  s(\Phi)\succ_{\log}|\lambda(\Phi)|\succ_{w\log}\lambda(\Phi_\text{SIC})=s(\Phi_{\rm SIC})\;.
  \label{eq:fortynineish}
\end{equation}
If $\{H_i\}$ and $\{\sigma_j\}$ are SICs, $\Phi^{-1}_{ij}=\frac{1}{d}\tr\Pi_i\Pi_j'$, where $\{\Pi_i\}$ and $\{\Pi_j'\}$ are SIC projectors in dimension $d$. Then
\begin{equation}
    \begin{split}
        &\left[\Phi^{-1}\Phi^{{-1}\dag}\right]_{ij}=\frac{1}{d^2}\sum_k(\tr\Pi_i\Pi_k')(\tr\Pi_j\Pi_k')\\
        &=\frac{1}{d^2}\tr\left[(\Pi_i\otimes\Pi_j)\left(\sum_k\Pi_k'\otimes\Pi_k'\right)\right]\\
        &=\frac{1}{d^2}\tr\left[(\Pi_i\otimes\Pi_j)\left(\frac{2d}{d+1}P_\text{sym}\right)\right]\\
        &=\frac{1}{d(d+1)}\tr\left[(\Pi_i\otimes\Pi_j)\left(I\otimes I+\sum_{kl}^d\ketbra{k}{l}\otimes\ketbra{l}{k}\right)\right]\\
        &=\frac{1+\tr\Pi_i\Pi_j}{d(d+1)}=\frac{d\delta_{ij}+d+2}{d(d+1)^2}=\left[\Phi_{\rm SIC}^{-2}\right]_{ij}\;,
    \end{split}
    \label{eq:fiftyish}
\end{equation}
where $P_{\rm sym}$ is the projector onto the symmetric subspace of $\mathcal{H}_d^{\otimes 2}$ and in the third step we employed the fact that the SICs form a minimal $2$-design \cite{Renes:2004}. This shows that the modulus of $\Phi$ is equal to $\Phi_{\rm SIC}$ and thus the singular values of $\Phi$ and $\Phi_\text{SIC}$ coincide. The product $\Phi^{{-1}\dag}\Phi^{-1}$ works out to be the same, by interchanging the roles of the two SICs.

On the other hand, suppose $s(\Phi)=s(\Phi_\text{SIC})$. The product of all the singular values is the absolute value of the determinant \cite{Horn:1994}, so $|\det\Phi^{-1}|=\det\Phi^{-1}_\text{SIC}\implies|\det\alpha|=1\implies\det G=\det G_{\rm SIC}\iff \{E_i\}$ is a SIC. Carrying through the consequences of the MIC being a SIC allows us to see from Eq.~\eqref{sicbound} that $\sigma_j$ is rank-$1$ because the upper bound is
saturated for SIC probability vectors. We may expand the $\{\sigma_j\}$ in the SIC projector basis,
\begin{equation}
    \sigma_j=\sum_k\left[\beta\right]_{jk}\Pi_k\;.
\end{equation}
Acting on both sides by a SIC POVM element and computing the trace of both sides, we see
\begin{equation}
    \left[\Phi^{-1}\right]_{ij}=\tr E_i\sigma_j=\sum_k\left[\beta\right]_{jk}\tr E_i\Pi_k=[\Phi^{-1}_{\rm SIC}\beta^T]_{ij}\;,
\end{equation}
so $|\det\Phi^{-1}|=|\det\Phi^{-1}_{\rm SIC}||\det\beta^T|=\det\Phi^{-1}_{\rm SIC}$ implies $|\det\beta|=1$. Denoting the Gram matrix of states by $g$, we have, in the same way as before,
\begin{equation}
    \det g=(\det\beta)^2\det g_{\rm SIC}=\det g_{\rm SIC}\;.
\end{equation}
We now prove that $\det g=\det g_{\rm SIC}$ implies that the basis of projectors forms a SIC. The following lemma is due to Huangjun Zhu~\cite{Zhu:2017}. We only use part of Zhu's conclusion, but the lemma is of enough interest to present in full.

\begin{lemma}[Zhu]\label{zhu}
    Let $\lambda$ be the spectrum of the Gram matrix $g$ of a normalized basis of positive semidefinite operators $\Pi_j$ sorted in nonincreasing order. Then $\lambda\succ\lambda_\textnormal{SIC}$ with equality iff $\Pi_j$ forms a SIC.
\end{lemma}
\begin{proof}
    By assumption $\text{tr}\Pi_j^2=1$ for all $j$. Since the eigenvalues of $\Pi_j$ are nonnegative,
    \begin{equation}
        1=\text{tr}\Pi_j^2=\sum_i\lambda^2_i(\Pi_j)\leq\sum_i\lambda_i(\Pi_j)=\text{tr}\Pi_j.
    \end{equation}
    Define the frame superoperator
    \begin{equation}
        \mathcal{F}=\sum_j\kketbra{\Pi_j}{\Pi_j}\;,
    \end{equation}
    where $\kket{A}:=\sum_{ij}[A]_{ij}\ket{i}\ket{j}$. $\mathcal{F}$ has the same spectrum as the Gram matrix $[g]_{ij}=\bbraket{\Pi_i}{\Pi_j}=\tr\Pi_i\Pi_j$. To see this, form a projector out of the state $\sum_i\kket{\Pi_i}\ket{i}$ where $\ket{i}$ is an orthonormal basis in $\mathcal{H}_{d^2}$ and perform partial traces over each subsystem. The results are $g^{\rm T}$ and $\mathcal{F}$, and so, by the Schmidt theorem, the spectra of $\mathcal{F}$ and $g$ are equal:
    $\lambda(g)=\lambda(\mathcal{F})=\lambda$.

    The expectation value of any operator
    with respect to an arbitrary normalized state is
    less than or equal to its maximal eigenvalue. Thus, a lower bound on the maximal eigenvalue $\lambda_1$ of $\mathcal{F}$ is given by
    \begin{equation}\label{evalbound}
        \lambda_1\geq\frac{1}{d}\bbra{I}\mathcal{F}\kket{I}=\frac{1}{d}\sum_j(\tr\Pi_j)^2\geq d.
    \end{equation}
    As our basis is normalized, $\tr g=d^2$, so $\sum_i\lambda_i=d^2$. With this constraint and our bound on the maximal eigenvalue, we have
    \begin{equation}
        \begin{split}
        \lambda&\succ\left(\lambda_1,\frac{d^2-\lambda_1}{d^2-1},\ldots,\frac{d^2-\lambda_1}{d^2-1}\right)\\
        &\succ\left(d,\frac{d}{d+1},\ldots,\frac{d}{d+1}\right)=\lambda_{\rm SIC}.
    \end{split}
    \end{equation}
    The second majorization becomes an equality when $\lambda_1=d$. From Eq.~\eqref{evalbound}, we can see that all $\Pi_j$ must be rank-$1$ for this condition to be satistfied. Furthermore, we see that in this case $\frac{1}{\sqrt{d}}\kket{I}$ is an eigenvector of $\mathcal{F}$ which achieves the maximal eigenvalue $d$. When both majorizations are equalities the spectrum $\lambda_\text{SIC}$ tells us that $\mathcal{F}$ takes the form of a weighted sum of a projector and the
    identity superoperator $\mathbf{I}$, specifically
    \begin{equation}
        \mathcal{F}=\frac{d}{d+1}\left(\mathbf{I}+\kketbra{I}{I}\right).
\end{equation}
By Cor.~1 in \cite{Appleby:2015}, this implies the $\Pi_j$ form a SIC.
\end{proof}
As in the Lemma, denote by $\lambda$ the spectrum of $g$ sorted in nonincreasing order. $\tr g=d^2$, so
\begin{equation}
    \sum_{i>1}\lambda_i=d^2-\lambda_1\;.
\end{equation}
Then because the arithmetic mean is greater than or equal to the geometric mean with equality iff the elements are all equal, we have
\begin{equation}
    \frac{1}{d^2-1}\sum_{i>1}\lambda_i=\frac{d^2-\lambda_1}{d^2-1}\geq\left(\prod_{i>1}\lambda_i\right)^{\!\frac{1}{d^2-1}}\;,
\end{equation}
which implies
\begin{equation}
    \det g\leq \lambda_1\left(\frac{d^2-\lambda_1}{d^2-1}\right)^{\!d^2-1}\;,
\end{equation}
with equality iff $\lambda_2=\cdots=\lambda_d^2=\frac{d^2-\lambda_1}{d^2-1}$. When $\lambda_1=d$, we then have
\begin{equation}
    \det g =\frac{d^{d^2}}{(d+1)^{d^2-1}}=\det g_{\rm SIC}\;,
\end{equation}
with equality iff $\lambda=\lambda_{\rm SIC}$. By Lemma \ref{zhu}, we have equality iff the post-measurement states form a SIC.
\label{appendix:theorem3}
\section{Appendix B: Proof of Theorem \ref{vol}}
Equation \eqref{rhoinbasis} expanded instead in the $\rho_i$ basis allows us to relate the differential elements of operator space and probability space for any MIC basis:
    \begin{equation}
        \text{d}\sigma=\sum_{i,j}[\Phi]_{ij}\rho_i\text{d}p^j\;.
    \end{equation}
The Hilbert--Schmidt line element is then
    \begin{equation}\label{messyline}
        \text{d}s^2_\text{HS}=\tr (\text{d}\sigma)^2=\sum_{ijkl}[\Phi]_{ij}[\Phi]_{kl}(\tr\rho_i\rho_k)\text{d}p^j\text{d}p^l\;.
    \end{equation}
    As in the proof of Lemma \ref{det}, we write $\Phi=AG^{-1}$ where $[G]_{ij}=\tr H_iH_j$ is the Gram matrix for the MIC and $[A]_{ij}=h_i\delta_{ij}$. Note further that $\tr\rho_i\rho_j=\left[A^{-1}GA^{-1}\right]_{ij}$. Then Eq.~\eqref{messyline} simplifies to
    \begin{equation}
        \text{d}s^2_\text{HS}=\sum_{ij}\left[G^{-1}\right]_{ij}\text{d}p^i\text{d}p^j\;.
    \end{equation}
    The Hilbert--Schmidt volume element on the space of Hermitian operators in $\mathcal{L}(\mathcal{H}_d)$ may now be related to the Euclidean volume element in $\mathbb{R}^{d^2}$,
    \begin{equation}
        \text{d}\Omega_\text{HS}=\sqrt{|\det G^{-1}|}\text{d}V_\text{E}\;,
    \end{equation}
    or, equivalently,
    \begin{equation}
        \text{d}V_\text{E}=\sqrt{|\det G|}\text{d}\Omega_\text{HS}\;.
    \end{equation}
    The larger $\det G$, the larger the corresponding Euclidean volume. Recall Eq.~\eqref{maxdet} which says
    \begin{equation}
        \det G\leq \det G_{\rm SIC}\;,
    \end{equation}
    with equality iff the MIC is a SIC. Thus, for any region in operator space, the Euclidean volume is maximal with respect to the SIC basis. In particular, the SIC basis gives the largest volume among positive semidefinite operators $A$ satisfying $1-\epsilon\leq\tr A\leq1+\epsilon$ for any $\epsilon>0$. As $\epsilon\to0$, we obtain quantum state space $\mathcal{Q}_d$ and the corresponding region in $\mathbb{R}^{d^2}$ will have the largest hyperarea within
$\Delta$ when computed with the SIC basis.

To calculate this hyperarea, we need to find the metric on $\Delta$ induced by the Hilbert--Schmidt metric in the SIC basis. We may parameterize $\Delta$ by
\begin{equation}
X=\left(p^1,\ldots,p^{d^2-1},1-\sum_{i=1}^{d^2-1}p^i\right)\;,
\end{equation}
which has partial derivatives $\partial_iX^\mu=\delta_i^\mu-\delta_{d^2}^\mu$ where the Latin index runs from $1$ to $d^2-1$ and the Greek index runs from $1$ to $d^2$. For any MIC, the induced metric $g$ is given by
\begin{equation}
    \left[g\right]_{ij}=\sum_{\mu,\nu=1}^{d^2}\partial_iX^\mu\partial_jX^\nu\left[G^{-1}\right]_{\mu\nu}\;.
\end{equation}
It is easily seen that $G^{-1}_\text{SIC}=d(d+1)I-J$ where $J$ is the Hadamard identity. One may then calculate
$g_\text{SIC}=d(d+1)(I+J)$
and $\det g_\text{SIC}=d^2(d^2+d)^{d^2-1}$. The induced volume element on $\Delta$ is then
\begin{equation}
    \text{d}\omega_\text{HS}=d\sqrt{(d^2+d)^{d^2-1}}\text{d}p^1\cdots \text{d}p^{d^2-1}\;.
\end{equation}
In a similar way, it may be checked that the Euclidean metric in $\mathbb{R}^{d^2}$ induces a volume element $\text{d}\mathcal{A}_\text{E}$ on $\Delta$ satisfying
\begin{equation}
    \frac{1}{d}\text{d}\mathcal{A}_\text{E}=\text{d}p^1\cdots \text{d}p^{d^2-1}\;,
\end{equation}
and so
\begin{equation}
    \text{d}\omega_\text{HS}=\sqrt{(d^2+d)^{d^2-1}}\text{d}\mathcal{A}_\text{E}\;.
\end{equation}
We may now integrate over quantum state space to obtain
\begin{equation}
    \text{vol}_\text{HS}(\mathcal{Q}_d)=\sqrt{(d^2+d)^{d^2-1}}\text{vol}_\text{E}(\mathcal{P}_\text{SIC})\;.
\end{equation}
\.Zyczkowski and Sommers \cite{Zyczkowski:2003} calculate the Hilbert--Schmidt volume of finite-dimensional quantum state space to be
  \begin{equation}
      \textnormal{vol}_\textnormal{HS}(\mathcal{Q}_d)=\sqrt{d}(2\pi)^{d(d-1)/2}\frac{\Gamma(1)\cdots\Gamma(d)}{\Gamma(d^2)}\;,
  \end{equation}
  from which Eq.~\eqref{volPsic} follows.
\end{document}